\documentclass[a4paper,12pt]{article}
\usepackage[margin=1.1in]{geometry}
\usepackage[backend=biber,doi=false,url=false,style=apa,sortcites=false]{biblatex}
\usepackage[titletoc, title]{appendix}
\usepackage[colorlinks,citecolor=blue,urlcolor=blue,linkcolor=blue]{hyperref}
\usepackage{tikz}
\usepackage{amssymb} 
\usepackage{kbordermatrix} 
\usepackage{amsmath} 
\usepackage{mathtools} 
\usepackage{amsthm} 
\usepackage{gensymb} 
\usepackage[skip=10pt,font=small,labelfont=bf]{caption} 
\usepackage{subcaption} 
\usepackage{enumitem} 
\usepackage{pdfpages}

\DeclareMathOperator*{\argmax}{arg\,max}

\DeclareMathOperator{\conv}{conv}
\DeclareMathOperator{\dis}{dis}

\theoremstyle{definition}
\newtheorem{axiom}{Axiom}
\newtheorem{proposition}{Proposition}

\newtheorem{lemma}{Lemma}
\newtheorem{theorem}{Theorem}
\newtheorem{assumption}{Assumption}

\newtheorem{axiomalt}{Axiom}[axiom]
\newenvironment{axiomp}[1]{
  \renewcommand\theaxiomalt{#1}
  \axiomalt
}{\endaxiomalt}

\newtheorem{assumptionalt}{Assumption}[assumption]
\newenvironment{assumptionp}[1]{
  \renewcommand\theassumptionalt{#1}
  \assumptionalt
}{\endassumptionalt}

\newcommand{\myx}{4}

\renewcommand{\kbldelim}{(}
\renewcommand{\kbrdelim}{)}

\makeatletter
\def\@fnsymbol#1{\ensuremath{\ifcase#1\or *\or\ddagger\or
   \mathsection\or \mathparagraph\or \|\or **\or \dagger \else\@ctrerr\fi}}
\makeatother

\title{Rational Bargaining:\\ Characterization and Implementation\thanks{I gratefully acknowledge financial support from the \textit{Jubiläumsfonds} of the Austrian National Bank (OeNB), project number 18719.}}
\author{Philipp Peitler\thanks{University of Vienna; philipp.peitler@univie.ac.at}}
\date{April 30, 2024 \\ \href{https://philipppeitler.github.io/mywebsite/bargaining.pdf}{[Latest version here]}}
\begin{document}
\maketitle

\begin{abstract}
The von Neumann-Morgenstern axioms are uncontroversial desiderata for individual decision-making. We say that a bargaining solution is rational if it can be interpreted as the most preferred alternatives under these axioms. Yet, neither the Nash nor the Kalai-Smorodinsky bargaining solution is rational in this sense. We formalize two consequences of rationality, namely that one can neither be strictly better off nor strictly worse off from randomizing over different actions. These two axioms, together with other standard axioms, characterize the relative utilitarian bargaining solution. We then implement this bargaining solution in sub-game perfect equilibrium.
\end{abstract}

\newpage
\section{Introduction}

Bargaining is prevalent in many economic settings, for instance, when both sides of a market are concentrated. Typical examples are bargaining over prices between up- and downstream firms, wage bargaining between firms and unions, and bargaining between health care providers and insurers. A common assumption in the applied economic literature is that the outcome of the bargaining process is given by the Nash bargaining solution \autocite{nash1950bargaining} or a generalization thereof \autocite{kalai1977nonsymmetric}.\footnote{Recent applications of the Nash solution for bargaining between up- and downstream firms can be found in \textcite{crawford2012welfare}, \textcite{shang2016information}, \textcite{crawford2018welfare}, \textcite{rogerson2020modelling} and \textcite{grunewald2023auto}, for wage bargaining in \textcite{cahuc2006wage}, \textcite{dobbelaere2018labor}, \textcite{depinto2019impact}, \textcite{piluso2023mobility} and \textcite{terai2023investment}, and for bargaining between insurers and health care providers in \textcite{gaynor2015industrial}, \textcite{gowrisankaran2015mergers}, \textcite{ho2017insurer}, \textcite{dafny2019price} and \textcite{ho2019equilibrium}.} While multiple bargaining protocols have been identified that support the Nash solution non-cooperatively,\footnote{See Section \ref{ch2:secimplementation} for implementations of the Nash bargaining solution.} the elegance of Nash's axiomatization undoubtedly plays a part in the popularity of the solution. 

Another setting in which bargaining is pervasive is the settlement of disputes outside of courts. Oftentimes, in order to facilitate agreement and prevent a costly legal battle, arbitrators or mediators are appointed to make a decision. The international division of the American Arbitration Association (AAA) alone handled over ten thousand cases in the year 2022, with close to 16 billion dollars in total claims.\footnote{See the 2022 AAA-ICDR B2B Case Statistics at \href{https://www.adr.org/research}{https://www.adr.org/research}. Accessed on November 5, 2023.} In the context of arbitration, Nash's axioms could be understood as desirable characteristics we would want an arbitrator to display.

Inherent to any economic decision is uncertainty about the final outcome. When up- and downstream firms negotiate prices, they face uncertainty about consumer demand. When firms and unions negotiate wages, they face uncertainty about future inflation rates. Both as a normative desideratum and as a benchmark, economic actors are typically assumed to deal with uncertainty rationally, meaning they act in accordance with the \textcite{vonneumann1944theory} axioms and maximize expected utility. Similarly, we feel that the collective decision of a group, such as the resolution of a bargaining problem, should be rational in this sense. If an arbitrator or mediator decides for the group, they should do so rationally as well. 

We define a bargaining solution as rational if it can be interpreted as the most preferred alternatives under some von Neumann-Morgenstern (vNM) preference relation. According to the vNM paradigm, a lottery is evaluated by first considering the value of each of the final outcomes and then taking a convex combination of these values. By reducing a lottery to the final outcomes, it is implicitly ruled out that the decision-maker values the process of randomization. As a consequence, a decision maker would only randomize over different outcomes if she were indifferent between them.\footnote{Consider for instance mixing in Nash equilibrium, which requires that players are indifferent between all actions over which they randomize.} The Nash bargaining solution violates this condition. To demonstrate this, consider an arbitrator who has to allocate a single indivisible item to either one of two agents. Let $a_i$ denote the allocation where Agent $i$ receives the item and let the arbitrator's preferences be captured by the utility function $u$. The Nash solution would prescribe that the arbitrator flips a coin and allocates the item to the winner, i.e., a 50-50 lottery between $a_1$ and $a_2$. The Nash solution selects neither $a_1$ nor $a_2$, over which the coin randomizes. Hence, it demands that the arbitrator strictly prefers the coin flip over the deterministic allocations. However, by the vNM axioms, the utility of the coin flip is given by $\frac{1}{2}u(a_1) + \frac{1}{2}u(a_2)$, which cannot exceed both $u(a_1)$ and $u(a_2)$. Therefore, the Nash solution is not rational. Note that the same is true for the bargaining solution by \textcite{kalai1975other}.

The incompatibility of rational risk preferences and ex-ante symmetry (i.e., choosing the coin-flip) is well known in the literature on preference aggregation \autocite{diamond1967cardinal, harsanyi1975nonlinear, broome1984uncertainty} and non-expected utility theory \autocite{machina1989dynamic}.\footnote{\textit{Ex-ante} refers to the fact that the coin-flip is only symmetric before the coin has landed.} Some authors claim that since we have to value ex-ante symmetry as a fairness principle, we need to give up rationality. We argue against this intuition later in this section. We believe that fairness can be understood as the arbitrator giving equal consideration to every agent, which doesn't necessitate that agents are made equal in their expected outcomes. It seems that the AAA shares a similar notion of fairness, as they write that the arbitrator should ``achieve a fair, efficient, and economical resolution of the dispute'', but also makes clear that they ``do not split the baby'' as they decide ``clearly in favor of one party in over 94.5\% of the cases''.\footnote{See the Commercial Arbitration Rules and Mediation Procedures at \href{https://www.adr.org/Rules}{https://www.adr.org/Rules}, pages 22, 23 and 26, and \href{https://go.adr.org/split-the-baby.html}{https://go.adr.org/split-the-baby.html}. Accessed on November 5, 2023.}

In this paper, we aim to find a rational bargaining solution. To capture a central aspect of rationality, we propose the \textit{no benefit of randomization} (\ref{ch2:NBR}) axiom. This axiom relates non-convex bargaining sets to their convex hull. A non-convex bargaining set represents a situation where randomization is either not feasible or not permitted. Allowing for randomization means the group or arbitrator can choose from the convex hull of this set. The axiom then says that if an alternative is selected by the bargaining solution in the non-convex set, it must still be selected in the convex hull of the set. Another consequence of the vNM axioms, and the flip-side of \ref{ch2:NBR}, is that the agent is never strictly worse off when randomizing. Hence, when randomization is possible (i.e., the bargaining set is convex) and two alternatives are selected by the bargaining solution, then any mixture (i.e., convex combination) of these two alternatives must be selected as well. We call this axiom \textit{convexity} (\ref{ch2:CS}). We find that \ref{ch2:NBR} and \ref{ch2:CS}, together with standard axioms, characterize the \textit{relative utilitarian} (RU) bargaining solution. The RU bargaining solution selects the alternatives with the highest sum of normalized utilities. Utilities are normalized such that the disagreement point has utility 0, and the best alternative  has utility 1. The other axioms that underlie our characterization are invariance to the utility-scale, strong Pareto, weak symmetry, and a weaker version of Nash's independence of irrelevant alternatives (IIA) axiom. The proof for two agents is simple and resembles the one of \textcite{nash1950bargaining}. The proof easily generalizes to any number of agents.

Besides a characterization, we also implement the RU bargaining solution in sub-game perfect equilibrium. This means that we identify a bargaining protocol, which, in equilibrium, leads to a RU-optimal alternative. We show that full implementation is not possible and identify a game form that weakly implements our bargaining solution. However, we are quite close to full implementation, as any RU-optimal alternative that is strictly better than the disagreement point for every agent is an equilibrium outcome. In this game, agents simultaneously make a proposal consisting of an alternative and a probability for each agent. In equilibrium, all agents propose the same RU-optimal alternative, and this alternative is implemented immediately. If there is disagreement among the agents, the proposal with the highest sum of probabilities has to be sequentially approved by all agents. Agents can choose whether to accept the alternative or receive a utility equal to the probability that the proposal assigns to them. Hence, the probabilities can be understood as a claim regarding how good the alternative is for each agent. If the claim was inflated, the alternative is rejected by at least one agent. If an agent rejects, all other agents receive a utility equal to the disagreement point.

Finally, we consider two other rational solutions. First, we provide the first characterization of the asymmetric relative utilitarian (ARU) solution. Similar to the asymmetric Nash solution by \textcite{kalai1977nonsymmetric}, the ARU solution generalizes the RU solution by allowing for different weights on the individuals' utilities. This can capture differences in bargaining power, which is important for applications. Second, we characterize the utilitarian solution. This solution maximizes the sum of individual utilities without normalizing them first. It is applicable when utilities have absolute meaning, for instance, in the case where utilities express individuals' willingness to pay to bring about a social alternative. The axiomatization requires only a small tweak to the axioms that characterize RU solution. We simply drop the invariance axiom and impose Nash's IIA in its full strength.

We now come back to the incompatibility of rationality and ex-ante symmetry. Arguments in favor of rational risk preferences are well known, so we do not recapitulate them here. Instead, we will argue that the desirability of the coin flip is, at least in part, due to reasons one is supposed to abstract from. First, in a bargaining situation, individuals might not only receive utility from the final outcome but might also care about the procedure by which this outcome is implemented. An individual might be better off losing a public coin flip compared to the item being allocated to the other agent directly because they care about being treated symmetrically. Furthermore, an individual might be better off winning a public coin flip, compared to the item being allocated to them directly, because they feel better about receiving the item when the other agent had a chance as well. If this is the case, then the procedure of flipping the coin publicly is different from simply mixing over the alternatives, for instance, by flipping the coin in secret. In fact, the former Pareto dominates the latter. A rational arbitrator cannot strictly prefer the secret coin flip to either of the direct allocations, but she could strictly prefer the procedure of flipping the coin publicly if individuals indeed care about procedures. In order to abstract from this possibility, one should think of randomization as being performed in secret. Second, flipping a coin seems desirable because it is the obvious tie-breaking rule when the arbitrator is indifferent between giving the item to either of the agents. However, similar to choice correspondences in the theory of individual decision-making, a set-valued bargaining rule simply does not make a statement about how ties are broken in case of multiplicity. It is perfectly compatible with rationality to impose such a tie-breaking rule as a second-order principle, but only after first-order principles have pinned down the solution. Third, the coin flip would prevent a biased arbitrator from giving the item to her favored agent. However, since we propose a normative theory to resolve the bargaining problem, the arbitrator is unbiased by assumption. To summarize, rationality is compatible with a strict preference for public randomization, with randomization as a matter of tie-breaking and with randomization to limit a biased arbitrator's ability to discriminate. Once we abstract from these aspects, are we still willing to give up rationality in favor of flipping the coin?

\subsection{Literature}

Other axiomatizations of the RU bargaining solution are by \textcite{pivato2009twofold} and \textcite{baris2018timing}. \textcite{pivato2009twofold} considers preferences over bargaining solutions and then imposes axioms on these preferences. This differs from the standard approach, established by \textcite{nash1950bargaining}, where axioms are imposed on the bargaining solution directly. \textcite{baris2018timing} adapts the characterization of the utilitarian bargaining solution by \textcite{myerson1981utilitarianism} to a utility-scale invariant setting. Their central axiom can be interpreted as a dynamic consistency condition. When facing uncertainty over what the bargaining set will be, the arbitrator makes a plan, which specifies for each possible bargaining set a utility vector. Then, the expected utility vector must be the solution in the expected bargaining set. \textcite{cao1982preference} identifies necessary axioms for the RU bargaining solution but does not provide a characterization. Note that \textcite{cao1982preference}, \textcite{pivato2009twofold} and \textcite{baris2018timing} assume the bargaining set to be convex, whereas we contribute to the literature on bargaining over non-convex sets \autocite{kaneko1980extension, zhou1997nash, mariotti1998extending, mariotti1998nash, conley1996extension, denicolo2000nash, ok1999revealed, nagahisa2002axiomatization, xu2006alternative, zambrano2016vintage}.

Related to the RU bargaining solution is a large literature on preference aggregation, which characterizes a relative utilitarian rule \autocite{karni1998impartiality, dhillon1999relative, segal2000let, borgers2017revealed, marchant2019utilitarianism, sprumont2019relative, brandl2021beliefaveraging, peitler2023putting, karni2024impartiality}. Especially related is \textcite{peitler2023putting}. In an application of their aggregation rule, the RU bargaining solution is derived from the most preferred element of a menu-dependent social preference, where the menu consists of the alternatives that are better for every agent than the disagreement point.

Related to rational bargaining is a literature on the rationalizability of bargaining rules \autocite{peters1991independence, bossert1994rational, sanchez2000rationality, xu2013rationality}.  A bargaining rule is rationalizable if it can be interpreted as the most preferred alternative under a single preference relation over utility vectors, which applies independently of the bargaining set. In this literature, rationality is understood as satisfying the weak axiom of revealed preference (or similar conditions). We take rationality to mean that the bargaining solution is consistent with the maximization of a vNM preference relation, but we do not insist that it is the same preference relation for every bargaining set.

Implementations of the RU bargaining solution have been provided by \textcite{miyagawa2002subgameperfect} and \textcite{hagiwara2020subgameperfect}. They, however, consider the case of only two agents and strictly convex bargaining sets. Note that under strictly convex bargaining sets, there is a unique RU-optimal alternative. Our game, on the other hand, can have multiple equilibrium outcomes, each corresponding to one of the multiple RU-optimal alternatives. Our implementation is similar to the ones by \textcite{moulin1984implementing} and \textcite{moore1988subgame}. \textcite{moulin1984implementing} implements the Kalai-Smorodinsky solution for convex bargaining sets. As in our implementation, a proposal has to be sequentially approved by all players. \textcite{moore1988subgame} provides general results on the implementability of social welfare functions in sub-game perfect equilibrium. The first stage of our game is similar to theirs, however, they require each individual to report the state (i.e., the entire utility function of every player), whereas we ask players to only report the utility for one alternative.

\section{Axiomatization}\label{ch2:secaxioms}

Let $N:=\{1,...,n\}$ be a set of agents where $n\in \mathbb{N}$ and $n\geq 2$. A bargaining problem $(S,d)$ consists of a bargaining set $S\subseteq \mathbb{R}^n$ and a disagreement point $d\in S$. We simplify notation by normalizing the disagreement point to $d=(0,...,0)$ and write $S$ instead of $(S,d)$. Let $\mathbb{R}_+$ denote the positive real numbers, including 0. We restrict attention to bargaining sets $S$ where (i) $S\subseteq \mathbb{R}^n_+$ (ii) $S$ is compact, and (iii) for each $i\in N$, there exists a $u\in S$ such that $u_i>0$. We denote the domain of bargaining sets that satisfy these properties by $\mathcal{S}$. A bargaining solution $f$ is a correspondence that assigns to every $S \in \mathcal{S}$ a non-empty subset of $S$.

Our central axiom is \textit{no benefit of randomization}. For any $R \subset \mathbb{R}^n$, let $\conv R$ denote the convex hull of $R$.

\begin{axiomp}{NBR}[No Benefit of Randomization]\label{ch2:NBR}
For every $S\in \mathcal{S}$, $$f(S) \subseteq f(\conv S).$$
\end{axiomp}

We illustrate the axiom with the help of the following example. Consider a finite bargaining set $S$, arising from the allocation of finitely many indivisible items. Now consider a lottery $l$ that realizes some allocation $v\in S$ with probability $\lambda\in(0,1)$ and another allocation $v'\in S$ with probability $1-\lambda$. Since the utilities in $v$ and $v'$ express individuals' vNM preferences, Agent $i$'s utility of $l$ is $\lambda v_i + (1-\lambda) v'_i$. Hence, if $l$ were feasible, it would be a point in the bargaining set at $\lambda v + (1-\lambda) v'$. See the left-hand side of Figure \ref{ch2:fig:nbr} for an illustration when $n=2$. \begin{figure}[h]
\centering
\begin{subfigure}{.5\textwidth}
  \centering
  \includegraphics[scale=1]{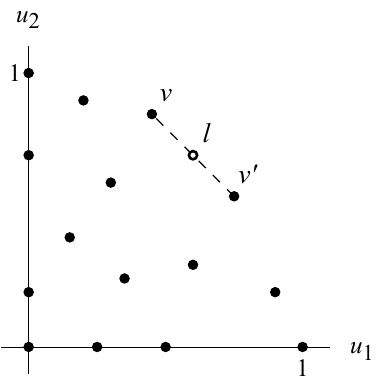}
  \caption{$S$.}
\end{subfigure}%
\begin{subfigure}{.5\textwidth}
  \centering
  \includegraphics[scale=1]{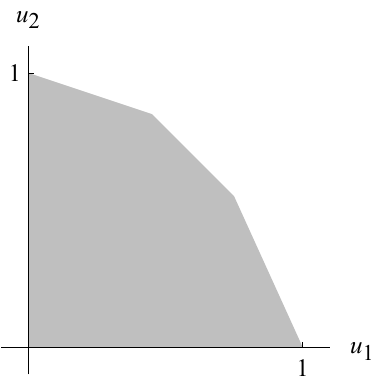}
  \caption{$\conv S$.}
\end{subfigure}
\caption{Convexification of the bargaining set.}
\label{ch2:fig:nbr}
\end{figure} Consequently, if every lottery over the allocations in $S$ were feasible, the bargaining set would be the convex hull of $S$. This case is depicted on the right-hand side of Figure \ref{ch2:fig:nbr}. \ref{ch2:NBR} says that if an allocation is optimal in the non-convex bargaining set $S$, where randomization isn't feasible, then this allocation must still be optimal in the convex bargaining set $\conv S$, where randomization is feasible. Hence, randomization doesn't introduce a lottery, strictly better for the group than the best allocation, as this would contradict that the group acts in accordance with the vNM postulates. Note, however, that randomization can introduce new utility vectors that are equally optimal as an allocation, which is why \ref{ch2:NBR} does not demand $f(S) = f(\conv S)$.

Note that \ref{ch2:NBR} has no bite when the domain is restricted to convex bargaining sets, which is the classic setting of \textcite{nash1950bargaining} and \textcite{kalai1975other}. However, these popular bargaining solutions have been extended in various ways to domains that include non-convex sets. We find that these extensions violate \ref{ch2:NBR}. In the following, we demonstrate this for the extensions by \textcite{xu2006alternative}. Let $f^{\text{Nash}}$ denote the Nash bargaining solution and $f^{\text{KS}}$ denote the Kalai-Smorodinsky bargaining solution as in \textcite{xu2006alternative}.

\begin{proposition} Both $f^{\text{Nash}}$ and $f^{\text{KS}}$ violate \ref{ch2:NBR}.
\end{proposition}

\begin{proof}
Assume $n=2$ and consider the barging set $$S=\{(u_1,u_1)\in [0,1]^2:u_1\leq x\text{ or }u_2\leq x\}$$ for some $x\in(0,1)$. Then $f^{\text{Nash}}(S)=\{(1,x),(x,1)\}$ and $f^{\text{KS}}(S)=\{(x,x)\}$. However, $f^{\text{Nash}}(\conv S)=f^{\text{KS}}(\conv S)=\{(\frac{1+x}{2},\frac{1+x}{2})\}$. Hence, \ref{ch2:NBR} is violated. Figure \ref{ch2:fig:prop} illustrates this for the Nash solution.
\end{proof}

\begin{figure}[h]
\centering
\begin{subfigure}{.5\textwidth}
  \centering
  \includegraphics[scale=1]{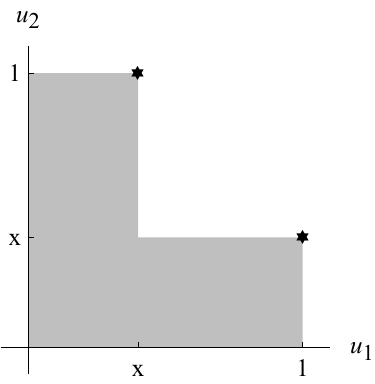}
  \caption{$S$.}
\end{subfigure}%
\begin{subfigure}{.5\textwidth}
  \centering
  \includegraphics[scale=1]{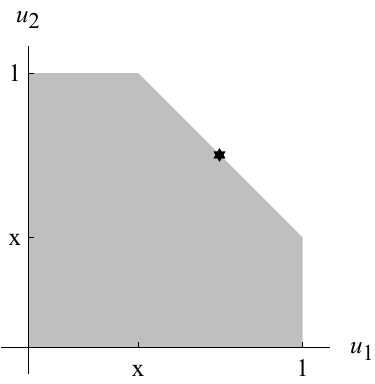}
  \caption{$\conv S$.}
\end{subfigure}
\caption{Violation of \ref{ch2:NBR} by $f^{\text{Nash}}$. Stars indicate the solution.}
\label{ch2:fig:prop}
\end{figure}

Another consequence of the vNM axioms is that the arbitrator cannot be strictly worse off under randomization. If two utility vectors are both optimal from the perspective of the arbitrator, then a lottery over these vectors, assuming it is feasible, must also be optimal. This is captured by the following axiom.

\begin{axiomp}{CONV}[Convexity]\label{ch2:CS}
For every $S \in \mathcal{S}$, $f(S)$ is convex whenever $S$ is convex.
\end{axiomp}

Note that both the Nash and Kalai-Smorodinsky bargaining solutions trivially satisfy this axiom since these solutions are singletons whenever the bargaining set is convex.

Since the prominent bargaining solutions violate \ref{ch2:NBR} and are therefore not rational, we are in need of an alternative solution. Besides the aforementioned axioms, this solution should satisfy agreed-upon desiderata. Both \textcite{nash1950bargaining} and \textcite{kalai1975other} agree that a solution should be Pareto efficient, invariant to the utility-scale, and that it should give equal treatment to symmetric agents. These axioms have been formulated under the assumption that the solution is single-valued. In the following, we translate these axioms to our setting, where the solution can be set-valued.

\begin{axiomp}{PO}[Pareto Optimality]\label{ch2:PO}
For every $S\in \mathcal{S}$, if $u\in f(S)$, then there is no $v\in S$ such that $v \neq u$ and $v_i \geq u_i$ for all $i\in N$.
\end{axiomp}

We say that $\alpha$ is a positive linear transformation if there exists $k_1,...,k_n > 0$ such that for any $R \subseteq \mathbb{R}^n$, $\alpha(R)=\{u\in \mathbb{R}^n:(k_1 u_1,...,k_n u_n) \in R\}$.

\begin{axiomp}{INV}[Invariance]\label{ch2:INV}
For every $S \in \mathcal{S}$ and positive linear transformation $\alpha$, $$f(\alpha(S))=\alpha(f(S)).$$
\end{axiomp}

For any $R \subseteq \mathbb{R}^n$, we say that $R$ is symmetric if, for every $u\in R$, any permutation of $u$ is in $R$ as well.

\begin{axiomp}{SYM}[Symmetry]\label{ch2:SYM}
For every $S\in \mathcal{S}$, if $S$ is symmetric and $u\in f(S)$, then $u=(x,...,x)$ for some $x\in\mathbb{R}$.
\end{axiomp}

Note that \ref{ch2:SYM} violates \ref{ch2:PO} in non-convex settings. To see this, consider $n=2$ and $S=\{(0,0),(1,0.9),(0.9,1)\}$. Then \ref{ch2:SYM} would demand $f(S)=\{(0,0)\}$, a violation of \ref{ch2:PO}. Therefore, for domains that include non-convex bargaining sets, symmetry needs to be weakened.

\begin{axiomp}{WSYM}[Weak Symmetry]\label{ch2:WSYM}
For every $S\in\mathcal{S}$, if $S$ is symmetric, then so is $f(S)$.
\end{axiomp}

For the fourth and final axiom, \textcite{nash1950bargaining} has independence of irrelevant alternatives (IIA) and \textcite{kalai1975other} has monotonicity. Note that IIA would not be compatible with the axioms we have imposed so far (\ref{ch2:NBR}, \ref{ch2:CS}, \ref{ch2:PO}, \ref{ch2:INV}, and \ref{ch2:WSYM}) and there is no obvious extension of monotonicity to set-valued solutions. However, there is a weaker version of IIA that is satisfied by both the Nash and Kalai-Smorodinsky solution, which we call \textit{weak IIA}.\footnote{Weak IIA and variations thereof have already appeared in \textcite{yu1973class}, \textcite{roth1977independence}, \textcite{cao1982preference}, \textcite{imai1983individual}, \textcite{dubra2001asymmetric}, \textcite{nagahisa2002axiomatization}, \textcite{xu2006alternative} and \textcite{rachmilevitch2019egalitarianism}.} Furthermore, this axiom replaces monotonicity in axiomatizations of the Kalai-Smorodinsky solution in non-convex settings \autocite{nagahisa2002axiomatization, xu2006alternative}. Since weak IIA is a common denominator of the Nash and Kalai-Smorodinsky solution, we will impose it as well. Before we introduce this axiom, let us first translate Nash's IIA to a setting with set-valued solutions.

\begin{axiomp}{IIA}[Independence of Irrelevant Alternatives]\label{ch2:IIA}
For every $S,S'\in \mathcal{S}$, if $S' \subset S$ and $f(S)\cap S'$ is non-empty, then $$f(S')=f(S)\cap S'.$$
\end{axiomp}

Next, we state weak IIA. For any $S\in \mathcal{S}$ and $i\in N$ let $m_i(S):=\max\{u_i:u\in S\}$. Furthermore, let $m(S):=\left(m_i(S)\right)_{i\in N}$. Following \textcite{yu1973class}, we call $m(S)$ the \textit{utopian point} of $S$.

\begin{axiomp}{WIIA}[Weak IIA]\label{ch2:WIIA}
For every $S,S'\in \mathcal{S}$, if $S' \subset S$, $m(S')=m(S)$ and $f(S)\cap S'$ is non-empty, then $$f(S')=f(S)\cap S'.$$
\end{axiomp}

\ref{ch2:IIA} says that removing points from the bargaining set does not change what is optimal from the perspective of the group. \ref{ch2:WIIA} weakens this condition, by imposing the former demand only in cases where the removal of points does not change the utopian point. The utopian point $m(S)$ is the maximal possible utility of each agent under the bargaining set $S$. While the utopian point typically isn't feasible (i.e., $m(S)\notin S$), it is an anchor point that allows us to relate the utilities of different agents. Since vNM utilities are unique only up to a positive affine transformation, comparisons of absolute utility levels across different individuals are meaningless. However, by \ref{ch2:INV} we can normalize every bargaining set such that $m(S)=(1,...,1)$. Then all utilities express how well-off each agent is relative to their best possible outcome. In contrast to absolute utility levels, this measure can be meaningfully compared across individuals. Removing a point $u\in S$ from $S$ such that $m_i(S\setminus\{u\})<m_i(S)$ for some $i\in N$ would require us to re-normalize the bargaining set, which would change how points other than $u$ are perceived by the group. Unlike \ref{ch2:IIA}, \ref{ch2:WIIA} allows this change of perspective to influence what is collectively optimal.

We show that the above axioms characterize the \textit{relative utilitarian} bargaining solution. Define $f^{\text{RU}}$ such that for every $S\in \mathcal{S}$, $$f^{\text{RU}}(S)=\argmax_{u\in S} \sum_{i\in N} \frac{u_i}{m_i(S)}.$$

\begin{theorem}\label{ch2:RU theorem}
$f$ satisfies \ref{ch2:NBR}, \ref{ch2:CS}, \ref{ch2:PO}, \ref{ch2:INV}, \ref{ch2:WSYM} and \ref{ch2:WIIA} if and only if $f\equiv f^{\text{RU}}$.
\end{theorem}

For $n=2$, the proof is short, simple, and resembles the one of \textcite{nash1950bargaining}. For this reason, we present it here in the main section. A proof for a general number of agents is in Appendix \ref{ch2:appA}.

\begin{proof}
First, consider the bargaining set $S_x=\{(0,0),(1,0),(0,1),(1,x),(x,1)\}$ for some $x\in[0,1]$, as illustrated in Figure \ref{ch2:fig:proof1}. \begin{figure}[h]
\center
\includegraphics[scale=1]{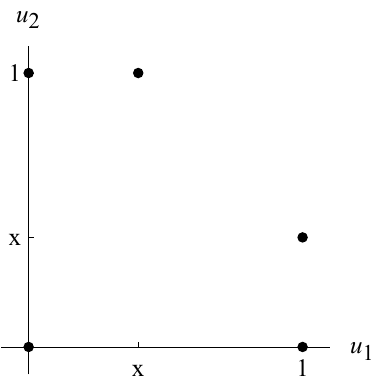}
\caption{$S_x$.} 
\label{ch2:fig:proof1}
\end{figure} Note that $P_x:=\{(1,x),(x,1)\}$ is the set of Pareto optimal points in $S_x$ and by \ref{ch2:PO}, $f(S_x)\subseteq P_x$. Furthermore, $S_x$ is symmetric. So by \ref{ch2:WSYM}, \begin{equation}\label{ch2:n2eq1}f(S_x)= P_x = \{(1,x),(x,1)\}.\end{equation}

Second, consider the bargaining set $\conv S_x$. By (\ref{ch2:n2eq1}) and \ref{ch2:NBR}, $P_x \subseteq f(\conv S_x)$. As $\conv S_x$ is convex, $\conv P_x\subseteq f(\conv S_x)$ by \ref{ch2:CS}. Note that $\conv P_x = \{u\in[0,1]^2:u_1+u_2 = 1+x\}$ and $\conv S_x=\{u\in[0,1]^2:u_1+u_2 \leq 1+ x\}$. Hence, $\conv P_x$ is the set of Pareto-optimal points in $\conv S_x$. Hence, by \ref{ch2:PO},
\begin{equation}\label{ch2:n2eq2}f(\conv S_x)=\left\{u\in[0,1]^2:u_1+u_2 = 1+x\right\}.\end{equation}
See Figure \ref{ch2:fig:proof2} for an illustration. The dashed line indicates the solution.

\begin{figure}[h]
\center
\includegraphics[scale=1]{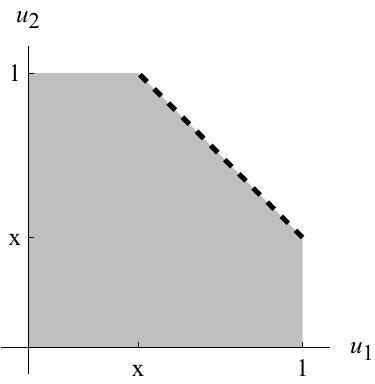}
\caption{$\conv S_x$.} 
\label{ch2:fig:proof2}
\end{figure}

Third, consider any bargaining set $S\in\mathcal{S}$ where $m(S)=(1,1)$. Let $x^*:=\max_{u\in S}(u_1+u_2)-1$, which must exist due to our assumption of compactness. Note that $S\subseteq \conv S_{x^*}$ and that $S\cap f(\conv S_{x^*})$ is non-empty. Hence, (\ref{ch2:n2eq2}) and \ref{ch2:WIIA} imply
\begin{equation}\label{ch2:n2eq3}f(S)=S\cap f(\conv S_{x^*})=\argmax_{u\in S} (u_1 + u_2).\end{equation}

Fourth, note that for any $S\in\mathcal{S}$ there exists a bargaining set $S' \in\mathcal{S}$ with $m(S')=(1,1)$ and a positive linear transformation $\alpha$ such that $\alpha(S)=S'$. Then by (\ref{ch2:n2eq3}) and \ref{ch2:INV}, \begin{equation}\label{ch2:n2eq4}f(S)=\argmax_{u\in S}\left(\frac{u_1}{m_1(S)}+\frac{u_2}{m_2(S)}\right).\end{equation}
\end{proof}

Above we have identified axioms that characterize the RU bargaining solution. Next, we show that these axioms are independent. We drop each of the axioms in Theorem \ref{ch2:RU theorem} and show that there is a solution, other than the RU solution, that satisfies the remaining axioms. 
\begin{enumerate}[label={(\arabic*)}]\setlength\itemsep{0em}
\item The Nash solution satisfies all axioms but \ref{ch2:NBR}.
\item Consider the solution that selects all maximal elements of the leximax preorder whenever $m(S)=(1,...,1)$. The solution for $m(S)\neq(1,...,1)$ is given by \ref{ch2:INV}. This solution satisfies all axioms but \ref{ch2:CS}.
\item $f(S)=\{(0,...,0)\}$ for all $S\in\mathcal{S}$ satisfies all axioms but \ref{ch2:PO}.
\item The utilitarian solution (Section \ref{ch2:sec:util}) satisfies all axioms but \ref{ch2:INV}.
\item The asymmetric relative utilitarian solution (Section \ref{ch2:sec:asym}) satisfies all axioms but \ref{ch2:WSYM}.
\item Let $\omega_i(S)=1+\max\{u_{i}:u\in S \text{ and }u_{i+1}=m_{i+1}(S)\}m_i(S)^{-1}$ if $i<n$ and $\omega_n(S)=1+\max\{u_{n}:u\in S \text{ and }u_{1}=m_{1}(S)\}m_n(S)^{-1}$. Then the solution $f(S)=\argmax_{u\in S} \sum_{i\in N} \omega_i(S) m_i(S)^{-1} u_i$ satisfies all axioms but \ref{ch2:WIIA}.
\end{enumerate}
This proves that the axioms are independent.

\section{Implementation}\label{ch2:secimplementation}

In the previous section, we have provided a normative theory on how a group should come to a solution in a bargain problem. However, how a group will arrive at a solution depends on the game form that describes the bargaining process. In the spirit of the Nash program \autocite{nash1953twoperson}, we implement the RU bargaining solution in sub-game perfect equilibrium (SPE), meaning we identify a game form where the RU bargaining solution arises as the SPE outcome. In line with the existing literature on implementation of the Nash bargaining solution \autocite{nash1953twoperson, rubinstein1982perfect, binmore1986nash, herrero1989nash, howard1992social, chae1994nperson, krishna1996multilateral, trockel2000implementations, miyagawa2002subgameperfect, trockel2002universal, guth2004nash, gomez2006achieving, britz2010noncooperative, okada2010nash, anbarci2013asymmetrica, britz2014convergence, abreu2015dynamic, qin2019implementation, hagiwara2020subgameperfect, hu2020bargaining, harstad2023pledgeandreview} and the Kalai-Smorodinsky bargaining solution \autocite{moulin1984implementing, trockel1999unique, miyagawa2002subgameperfect, haake2009two, anbarci2011nash, hagiwara2020subgameperfect}, we assume that the players have full information, i.e., know the utility functions of the other players.\footnote{See \textcite{serrano2005fifty, serrano2021sixtyseven} for a review of the literature on the Nash program.}

The bargaining protocol we propose can be used by an arbitrator who wants to bring about a desirable outcome but does not know the utility functions of the agents. For applicability in actual bargaining situations, it is important that the game is simple and intuitive. For instance, we would not be satisfied with a game where individuals have to report the state of the world, i.e., the entire utility function of every player, as in \textcite{moore1988subgame}. Note that, as in the previous section, we allow for non-convex bargaining problems, such that multiple alternatives can be optimal. This complicates the game somewhat, as players have to coordinate on one of multiple equilibria.

In the next section, we outline the basic setting.

\subsection{Preliminaries}

Let $A$ be a set of alternatives. We designate one alternative in $A$ as the disagreement alternative and denote it by $a_{\dis}$. Let $\Theta$ be a set of states. For every $\theta\in\Theta$ and $i\in N$, let $u_i^\theta:A\to[0,1]$ denote Player $i$'s vNM utility function over $A$, normalized such that $u_i^\theta(a_{\dis})=0$ and $\max_{a\in A} u_i^\theta(a)=1$. For every $\theta\in\Theta$, let $S_\theta$ denote the bargaining set associated with $\theta$, i.e., $S_\theta:=\{(u_1^\theta(a),...,u_n^\theta(a)):a\in A\}$. We assume that for each $\theta\in\Theta$, $S_\theta \in \mathcal{S}$. For any $\theta \in \Theta$, we say that $a\in A$ is \textit{RU-optimal} under $\theta$ if $(u_1^\theta(a),...,u_n^\theta(a))\in f^{\text{RU}}(S_\theta)$. For a game form $g$ and any $\theta\in \Theta$, we denote by $(g,\theta)$ the game with the game form $g$ and players' preferences according to $u_1^\theta$ to $u_n^\theta$. We say that a game form $g$ \textit{fully implements} the RU bargaining solution in SPE if, for every $\theta\in\Theta$ and $a\in A$, $a$ is RU-optimal under $\theta$ if and only if $a$ is an SPE outcome of $(g,\theta)$. Unfortunately, full implementation is not possible, which we discuss in Section \ref{ch2:sec:fullimp}. Instead, we fully implement the subset of RU-optimal alternatives that are strictly better than $a_{\dis}$ for every player. We call these alternatives \textit{strictly RU-optimal}. Note that this \textit{weakly implements} the RU bargaining solution, meaning every SPE outcome is RU-optimal.

A strictly RU-optimal alternative does not always exist. Consider, for example, the case of allocating a single indivisible item without randomization. Giving the item to any of the agents is RU-optimal, but not strictly RU-optimal, as none of the remaining agents improves relative to the disagreement alternative. In order to implement the strictly RU-optimal alternatives, we must ensure that at least one such alternative exists. Hence, we impose the following restriction on $\Theta$.

\begin{assumptionp}{A1}\label{ch2:A2}
For every $\theta\in\Theta$, there exists an $a\in A$ such that $a$ is strictly RU-optimal under $\theta$.
\end{assumptionp}

An additional assumption is imposed for convenience. We assume that for each agent there exists an alternative that is among the best for this agent and equal to the disagreement point for the remaining agents.

\begin{assumptionp}{A2}\label{ch2:A3}
For every $\theta\in\Theta$ and $i\in N$, there exists a $b^\theta_i\in A$ such that $u_i^\theta(b^\theta_i)=1$ and $u_j^\theta(b^\theta_i)=0$ for all $j\in N\setminus\{i\}$.
\end{assumptionp}

\ref{ch2:A3} holds true for many bargaining situations. For example, when bargaining over a surplus or over the allocation of goods, awarding everything to Player $i$ would correspond to the alternative $b^\theta_i$.

Next, we present the game form.

\subsection{The Game}\label{ch2:sec:game}

We now define the game form that fully implements the set of strictly RU-optimal outcomes. We denote this game form by $g^*$. The game form has two stages, an initial stage and, depending on the actions in the initial stage, an approval stage.

\textit{Initial Stage:} Each player simultaneously makes a proposal $(a,p)$, consisting of an alternative $a\in A$ and a list of $n$ strictly positive probabilities $p \in (0,1]^n$. We distinguish three cases:
\begin{enumerate}
\item [(1)] Assume all players make the same proposal $(a,p)$. Then $a$ is selected.
\item [(2)] Assume there are exactly two distinct proposals $(a,p)$ and $(a',p')$.
\begin{enumerate}
\item [(2a)] If $\sum_{i\in N} p_i = \sum_{i\in N} p'_i$, then $a_{\text{dis}}$ is selected.
\item [(2b)] If $\sum_{i\in N} p_i \neq \sum_{i\in N} p'_i$, then go to the \textit{approval stage}.
\end{enumerate}
\item [(3)] Assume there are three or more distinct proposals. Then the alternative of the proposal with the highest sum of probabilities that lies below $n$ is selected. In case of a tie, one of those proposals is chosen at random.
\end{enumerate}

\textit{Approval Stage:} For the approval stage to be reached, there must be exactly two distinct proposals $(a,p)$ and $(a',p')$. Assume without loss of generality that $\sum_{i\in N} p_i > \sum_{i\in N} p'_i$. Assign players into a sequence such that those who proposed $(a',p')$ come before those who proposed $(a,p)$. Any such sequence will do. According to this sequence, let players sequentially decide between \textit{accept} and \textit{reject}. If all players accept, then $a$ is selected. If Player $i$ rejects, then with probability $p_i$, Player $i$ can choose an alternative and with probability $1-p_i$, $a_{\text{dis}}$ is selected.

Figure \ref{ch2:fig:tree} illustrates the approval stage for three players, where Player 1 and 3 proposed $(a,p)$ and Players 2 proposed $(a',p')$.

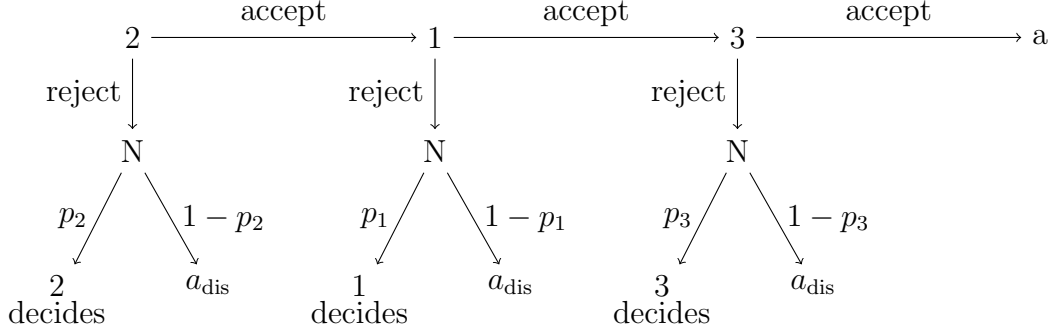
\begin{figure}[h]
\centering
\begin{tikzpicture}
  \node (node1) at (0*\myx, 0) {2};
  \node (node2) at (1*\myx, 0) {1};
  \node (node3) at (2*\myx, 0) {3};
  \node (nodefinal) at (3*\myx, 0) {a};
  
  \node (below1) at (0*\myx, -1.5) {N};
  \node (below2) at (1*\myx, -1.5) {N};
  \node (below3) at (2*\myx, -1.5) {N};
  
  \node[align=center, anchor=north] (dis1) at (0*\myx + 0.25*\myx, -3) {$a_{\text{dis}}$};
  \node[align=center, anchor=north] (dis2) at (1*\myx + 0.25*\myx, -3) {$a_{\text{dis}}$};
  \node[align=center, anchor=north] (dis3) at (2*\myx + 0.25*\myx, -3) {$a_{\text{dis}}$};
  
  \node[align=center, anchor=north, execute at begin node=\setlength{\baselineskip}{2ex}] (cho1) at (0*\myx - 0.25*\myx, -3) {2 \\ decides};
  \node[align=center, anchor=north, execute at begin node=\setlength{\baselineskip}{2ex}] (cho2) at (1*\myx - 0.25*\myx, -3) {1 \\ decides};
  \node[align=center, anchor=north, execute at begin node=\setlength{\baselineskip}{2ex}] (cho3) at (2*\myx - 0.25*\myx, -3) {3 \\ decides};
  
  \draw[->] (node1) -- (below1) node[midway, left] {reject};
  \draw[->] (node2) -- (below2) node[midway, left] {reject};
  \draw[->] (node3) -- (below3) node[midway, left] {reject};
  
  \draw[->] (node1) -- (node2) node[midway, above] {accept};
  \draw[->] (node2) -- (node3) node[midway, above] {accept};
  \draw[->] (node3) -- (nodefinal) node[midway, above] {accept};
  
  \draw[->] (below1) -- (dis1) node[midway, right] {$1-p_2$};
  \draw[->] (below2) -- (dis2) node[midway, right] {$1-p_1$};
  \draw[->] (below3) -- (dis3) node[midway, right] {$1-p_3$};
  
  \draw[->] (below1) -- (cho1) node[midway, left] {$p_2$};
  \draw[->] (below2) -- (cho2) node[midway, left] {$p_1$};
  \draw[->] (below3) -- (cho3) node[midway, left] {$p_3$};
\end{tikzpicture}
\caption{Example of approval stage.}
\label{ch2:fig:tree}
\end{figure}

Now that we have described the game $g^*$, we can state the following theorem.

\begin{theorem}\label{ch2:theoremimpl}
Assume \ref{ch2:A2} and \ref{ch2:A3} are satisfied. Then $g^*$ fully implements the strictly RU-optimal alternatives. Formally, for every $\theta\in\Theta$ and $a\in A$, $a$ is an SPE outcome of $(g^*,\theta)$ if and only if $a$ is strictly RU-optimal under $\theta$.
\end{theorem}

We sketch the proof here and provide a formal proof of the theorem in Appendix \ref{ch2:app:impl}. We fix some $\theta \in \Theta$ and omit it from now on. First, consider the approval stage and let $(a,p)$ denote the proposal with the higher sum of probabilities. Note that Player $i$'s expected utility of rejecting is $p_i \max_{b\in A} u_i(b) + (1-p_i) u_i(a_\text{dis})=p_i$. Hence, if $p_i > u_i(a)$ for all $i\in N$, then the unique SPE is that all players accept and $a$ is implemented. However, if $p_i < u_i(a)$ for some $i\in N$, then some player will reject. Hence, one can interpret the vector $p$ as a \textit{report} about the utility of $a$ for each player and the approval stage as a test of whether this report was truthful. We say that the report $p$ of a proposal $(a,p)$ is \textit{inflated} if $\sum_{i\in N} p_i > \sum_{i\in N} u_i(a)$ and we say that it is \textit{truthful} if $p=u(a)$.

Next, consider the initial stage and assume that all players propose $(a,u(a))$ for some strictly RU-optimal alternative $a$. We show that no deviation $(a',p')$ is profitable for Player $j$. There are three options for Player $j$, to report a higher sum, an equal sum or a lower sum of probabilities. Reporting an equal sum would lead to $a_\text{dis}$. Reporting a higher sum would mean that $a'$ is put to the test in the approval stage and that $j$ will decide last. However, because there is no alternative with a higher sum of utilities than $a$, $p'$ is necessarily inflated and $a'$ will be rejected by some player before $j$. Reporting a lower sum would mean that $a$ is put to the test in the approval stage and that $j$ will decide first. However, since the report of the proposal $(a,u(a))$ was truthful, there is an equilibrium where all accept, and $a$ is implemented anyway. In conclusion, there is no profitable deviation for Player $j$.

Conversely, consider the case where all players propose $(a,p)$, but $a$ isn't strictly RU-optimal. If $p$ is inflated, then there must be one Player $j$ for whom $p_j>u(a)$. This player can then put $(a,p)$ to the test by deviating to a report with a lower sum. Player $j$ is allowed to decide first and rejects. If $p$ is not inflated, then there is some player who would prefer a strictly RU-optimal alternative $a'$ to $a$. This player can propose $a'$ and report a probability for each player that is slightly below the true utility of $a'$. Then $a'$ is put to the test in the approval stage and every player accepts.

Above, we have shown that a non-optimal alternative cannot arise in equilibrium through a unanimous proposal, i.e., Case (1). However, we have yet to show that such an alternative cannot arise in equilibrium via Cases (2b) or (3). Note that we have designed Case (3) similar to the integer game in \textcite{moore1988subgame}, such that it cannot arise in equilibrium. A player can always ``outbid" the others by choosing an even higher sum below $n$. Furthermore, with three or more players, one can deviate from Case (2b) to induce Case (3) and implement their most preferred alternative. Showing that no sub-optimal equilibrium exists for two players is more involved and we leave this to the formal proof in the appendix.

\subsection{Full Implementation}\label{ch2:sec:fullimp}

In this section, we discuss the case of full implementation. The following proposition shows that it is not possible to fully implement the RU bargaining solution.

\begin{proposition}\label{ch2:prop:full}
Let $n=2$ and assume that for every $S\in\mathcal{S}$ there exists a $\theta\in\Theta$ such that $S_\theta=S$. Then there doesn't exist an extensive game form $g$ that fully implements the RU bargaining solution.
\end{proposition}

A formal proof is in Appendix \ref{ch2:app:full}. The intuition for this result goes as follows. Consider a two-agent bargaining problem $S\in\mathcal{S}$ where both $(1,0)$ and $(0,1)$ are in $f^{\text{RU}}(S)$. An example of such a problem is the division of a dollar among risk-neutral agents, where any efficient division is RU-optimal, including allocating the entire dollar to one of the agents. Now consider a state $\theta$ such that $S_\theta=S$ and a game form $g$ that implements the solution. Then there must exist two strategy profiles $s^+,s^-$ that are SPE of $(g,\theta)$ and where $u^\theta(s^+):=(u_1^\theta(s^+),u_2^\theta(s^+))=(1,0)$ and $u^\theta(s^-)=(0,1)$. Since by assumption, 0 is the minimal utility for both players, $u^\theta_1(s_1,s_2^-)=0$ for all $s_1$ and $u^\theta_2(s^+_1,s_2)=0$ for all $s_2$. This, in turn, implies that $(s^+_1,s^-_2)$ is a Nash equilibrium with pay-offs $(0,0)$. In the formal proof, we then use $(s^+_1,s^-_2)$ to construct an SPE with pay-offs $(0,0)$. Since, $(0,0)$ is not in $f^{\text{RU}}(S)$, the RU bargaining solution cannot be fully implemented.

We feel that the division of a dollar among two risk neutral agents is a canonical problem that the implementation should be able to address. Hence, we did not want to rule out such problems, for instance by assuming that the bargaining set is strictly convex as in \textcite{miyagawa2002subgameperfect} and \textcite{hagiwara2020subgameperfect}. Instead, we decided to exclude solutions that assign a pay-off of 0 to one of the players.

\section{Other Rational Solutions}\label{ch2:sec:other}

\subsection{Asymmetric Relative Utilitarian Solution}\label{ch2:sec:asym}

In Section \ref{ch2:secaxioms} we have imposed a symmetry axiom as a normative requirement to treat all agents of the group fairly. In applications, however, we might want to take into account that individuals have different bargaining power. This can lead, in otherwise symmetric situations, to asymmetric outcomes. It is, therefore, of interest to generalize a given bargaining solution to a $n-1$ parameter family of solutions, where each parameter is a weight on an individual's utility, representing their bargaining power. Asymmetric generalizations have been provided for the Nash solution \autocite{harsanyi1972generalized,kalai1977nonsymmetric} and the Kalai-Smorodinsky solution \autocite{dubra2001asymmetric}. In the following, we provide the first generalization of the relative utilitarian solution. We say that $f$ is an \textit{asymmetric relative utilitarian solution} if there exists $(\mu_1,...,\mu_n)\in(0,1)^n$ with $\sum_{i\in N} \mu_i = 1$ such that for every $S\in\mathcal{S}$, 
$$f(S)= \argmax_{u\in S} \sum_{i\in N} \mu_i\frac{u_i}{m_i(S)}.$$ We denote this solution by $f^{\text{ARU}}$. Unfortunately, for the characterization of $f^{\text{ARU}}$ it doesn't suffice to merely drop the symmetry axiom.

\begin{proposition}\label{ch2:prop:asym} There exists an $f$ such that $f\neq f^{\text{ARU}}$ and $f$ satisfies \ref{ch2:NBR}, \ref{ch2:CS}, \ref{ch2:PO}, \ref{ch2:INV} and \ref{ch2:WIIA}.
\end{proposition}

We prove the proposition by providing two solutions as counterexamples. Each of these solutions can be ruled out by an additional axiom. These two additional axioms, together with \ref{ch2:NBR}, \ref{ch2:CS}, \ref{ch2:PO}, \ref{ch2:INV} and \ref{ch2:WIIA} then characterize $f^{\text{ARU}}$. For ease of exposition, we provide these counterexamples for $n=2$. 

The first counterexample is the solution $f'$, which selects among the (symmetric) relative utilitarian outcomes the one that most benefits Agent 1, formally $$f'(S)=\argmax_{u\in f^{\text{RU}}(S)} u_1.$$ First, we show that the axioms, as stated in Proposition \ref{ch2:prop:asym}, are satisfied. It is easy to see that \ref{ch2:PO}, \ref{ch2:INV} and \ref{ch2:WIIA} are satisfied. Since the solution is a singleton, \ref{ch2:CS} is trivially satisfied. While convexification can increase the set of relative utilitarian outcomes, it cannot change the best relative utilitarian outcome of each agent. Hence, \ref{ch2:NBR} is satisfied as well. Next, we identify an axiom that rules out such a solution. Note that the solution is discontinuous, in the sense of vNM continuity. To see this, consider the following example, illustrated by Figure \ref{ch2:fig:asym2}. \begin{figure}[h]
\center
\includegraphics[scale=1]{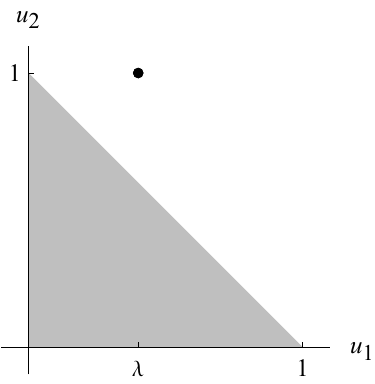}
\caption{Division of a dollar, with the risky technology at $(\lambda,1)$.} 
\label{ch2:fig:asym2}
\end{figure} An arbitrator has to divide a dollar among two risk-neutral agents. In addition, the arbitrator could choose to invest the dollar into a risky technology, which gives one dollar to Agent 2 for sure and with probability $\lambda$ one dollar to Agent 1. If $\lambda=1$, then the technology is Pareto-dominant and selected by $f'$. If $\lambda=0$, the technology is just as good as giving the dollar to Agent 2. Under $f'$, the arbitrator strictly prefers to give the dollar to Agent 1. Continuity would require that for some $\lambda$, the arbitrator is indifferent between the technology and giving the dollar to Agent 1, such that both are selected by the solution. However, such a $\lambda$ does not exist under $f'$, as the solution uniquely selects the technology for every $\lambda>0$. We impose the following axiom to ensure that a bargaining solution is continuous.

\begin{axiomp}{C}[Continuity]\label{ch2:C}
For every $S \in \mathcal{S}$ and $u\in S \setminus f(S)$, there exists a $\lambda\in [0,1]$ such that $$f(S\cup\{\lambda m(S) + (1-\lambda)u \}) = f(S) \cup \{\lambda m(S) + (1-\lambda) u \}.$$
\end{axiomp}

The axiom generalizes our previous example. Consider an arbitrary bargaining set $S$ and some point $u$ in $S$ that is not selected by the solution. Add a convex combination $\lambda m(S) + (1-\lambda)u$ between $u$ and the utopian point $m(S)$ to the original bargaining set. For $\lambda=1$, any solution satisfying \ref{ch2:PO} must select the convex combination, as it is equal to the utopian point. For $\lambda=0$, the convex combination is equal to $u$ and is therefore not selected by the solution. Axiom \ref{ch2:C} then states that for some $\lambda$ between 0 and 1, the convex combination must be in the solution, together with the solution of the original bargaining set $S$.

The second counterexample is a solution that can be described by linear, but non-parallel indifference curves. See Figure \ref{ch2:fig:asym} for an illustration. \begin{figure}[h]
\center
\includegraphics[scale=1]{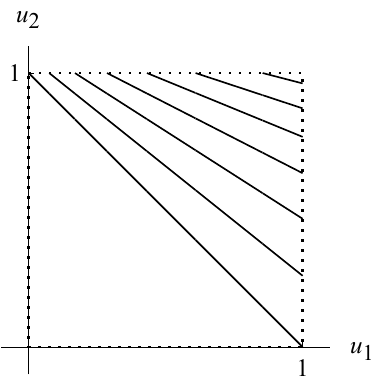}
\caption{Non-parallel indifference curves.} 
\label{ch2:fig:asym}
\end{figure} The solution selects the utility vectors on the highest indifference curve. These indifference curves are the same across all $S\in\mathcal{S}$ with $m(S)=(1,1)$. Solutions for bargaining sets with $m(S)\neq(1,1)$ are derived from \ref{ch2:INV}. We call this solution $f''$. \ref{ch2:INV} is then satisfied by assumption. It is easy to see that \ref{ch2:PO} and \ref{ch2:WIIA} are satisfied. Finally, \ref{ch2:CS} and \ref{ch2:NBR} are satisfied because indifference curves are linear. Note that the fanning-out of indifference curves is reminiscent of weighted expected utility \autocite{chew1979alphanu, chew1983generalization}. Such risk preferences violate the vNM independence axiom. Similarly, we can rule out non-parallel indifference curves through an independence axiom.

\begin{axiomp}{I}[Independence]\label{ch2:I}
For every $S \in \mathcal{S}$, if $u,v \in f(S)$, then for any $x\in [0,1]$, $$\{x m(S) + (1-x) u, x m(S) + (1-x) v \}\subseteq f(S \cup \{x m(S) + (1-x) u, x m(S) + (1-x) v \}).$$
\end{axiomp}

The axiom captures vNM independence in the bargaining setting. Consider any bargaining set, where the solution contains at least two points $u$ and $v$. It is as if the group was indifferent between these points. Now add two lotteries $l_u$ and $l_v$ that with probability $x$ give the utopian point $m(S)$ and otherwise $u$, in case of $l_u$, or $v$, in case of $l_v$. The possibility of the irrelevant alternative $m(S)$ in both $l_u$ and $l_v$ does not change the relative desirability between the two. Hence, both $l_u$ and $l_v$ must be in the solution of the new bargaining set.

Imposing axioms \ref{ch2:C} and \ref{ch2:I} in addition, leads to the asymmetric relative utilitarian bargaining solution.

\begin{theorem}\label{ch2:theo:asym}
$f$ satisfies \ref{ch2:NBR}, \ref{ch2:CS}, \ref{ch2:PO}, \ref{ch2:INV}, \ref{ch2:WIIA}, \ref{ch2:I} and \ref{ch2:C} if and only if $f\equiv f^{\text{ARU}}$.
\end{theorem}

We prove the theorem in Appendix \ref{ch2:appD}.

\subsection{Utilitarian Solution}\label{ch2:sec:util}

In the conventional understanding of the bargaining context, utilities derive from individuals' vNM preferences over the available alternatives. Since a utility representation of a vNM preference relation is unique only up to a positive affine transformation, the invariance axiom ensures that the solution is insensitive to the choice of the representation. However, in alternative scenarios, the utility scale might convey information. For instance, utilities might represent an individual's willingness to pay to bring about a given social alternative. Consider such a setting and think of the two agents who bargain over a single indivisible item. If Agent 1 values the item at \$100 and Agent 2 at \$50, it seems reasonable to award the item to Agent 1. Under the invariance axiom, however, this bargaining problem should be treated identically to one where both value the item at \$50. In order to take the absolute scale of valuations into account, \ref{ch2:INV} must be dropped. Previously we argued that \ref{ch2:IIA} imposes too stringent demands on the solution when utility scales lack significance. However, if utility scales are meaningful, \ref{ch2:IIA} can be assumed in its full strength. Adopting these two modifications leads to the \textit{utilitarian bargaining solution} \autocite{myerson1981utilitarianism}.

\begin{theorem}
$f$ satisfies \ref{ch2:NBR}, \ref{ch2:CS}, \ref{ch2:PO}, \ref{ch2:WSYM} and \ref{ch2:IIA} if and only if for every $S\in \mathcal{S}$, $$f(S)=\argmax_{u\in S} \sum_{i\in N} u_i.$$
\end{theorem}

\begin{proof}
Follow the first three steps of the proof of Theorem \ref{ch2:RU theorem} to find that $f(S)=\argmax_{u\in S} \sum_{i\in N} u_i$ whenever $m(S)=(1,...,1)$. Note that an analogous argument can be made when $m(S)=(x,...,x)$ for any $x>0$. Then by \ref{ch2:IIA}, $f(S)=\argmax_{u\in S} \sum_{i\in N} u_i$ even if $m(S)\neq(x,...,x)$.
\end{proof}

Note that if utilities are indeed valuations, the utilitarian sum is identical to economic surplus, a ubiquitous measure of welfare. Hence, a group that bargains rationally will also bargain welfare-optimally in the typical sense. Conversely, our axiomatization provides justification for the use of economic surplus as a welfare measure.

\section{Conclusion}\label{ch2:secconclusion}

This paper considers the implications of rational risk preferences on fair bargaining. While the vNM axioms are formulated in a setting where the primitives are preferences over alternatives, the canonical formulation of the bargaining problem by \textcite{nash1950bargaining} considers individuals' utilities directly. It is, therefore, not obvious what it means for a bargaining solution to deal with risk in a rational manner. We propose two axioms, \ref{ch2:NBR} and \ref{ch2:CS}, which capture the consequences of vNM preferences. Specifically, these axioms describe the role of randomization under the vNM theory. Together with standard axioms from the bargaining literature, these axioms lead to the relative utilitarian bargaining solution.

One might remark that in light of \textcite{harsanyi1955cardinal}, who shows that a rational group preference must be utilitarian, this result comes as no surprise. Note, however, that \ref{ch2:NBR} and \ref{ch2:CS} are necessary but not sufficient criteria for rational behavior under risk. In fact, they are quite far from being sufficient, as shown by the solutions that become permissible when \ref{ch2:WSYM} is dropped (Section \ref{ch2:sec:asym}). We, therefore, find it quite remarkable that these two rather weak axioms already pin down the relative utilitarian solution in the presence of standard bargaining axioms.

We believe that the relative utilitarian bargaining solution deserves a place in classrooms and textbooks, next to the solutions by \textcite{nash1950bargaining} and \textcite{kalai1975other}. As we have shown, the axioms that characterize the RU solution are straightforward, and the proof is short and simple. Furthermore, like the Nash and Kalai-Smorodinsky solution, the RU solution has a neat geometric interpretation. While the Kalai-Smorodinsky solution is the intersection of the Pareto frontier and the $45 \degree$ line and the Nash solution maximizes the area of the rectangle spanned by the origin and the Pareto frontier, the RU solution maximizes the circumference of the rectangle spanned by the origin and the Pareto frontier. We illustrate this with an example depicted by Figure \ref{ch2:fig:barg}. \begin{figure}[h]
\center
\includegraphics[scale=1]{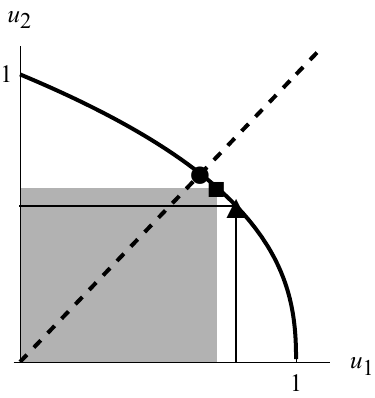}
\caption{The Nash, Kalai-Smorodinsky and RU solution.} 
\label{ch2:fig:barg}
\end{figure} The Pareto frontier is given by $u_2 = \left(1-u_1\right)^{\frac{2}{5}}$. The Kalai-Smorodinsky solution is indicated by a circle, the Nash solution by a square and the RU solution by a triangle.

We continue our discussion on the normative desirability of randomization, particularly when allocating an indivisible item between two agents. As noted in the introduction, randomization may seem appealing for reasons we are supposed to abstract from: individuals care about procedures, arbitrators might be biased, and randomization can be used to break ties. To disregard these factors, consider the following scenario:

(i) Individuals care only about the final outcomes and are completely indifferent to the procedures by which these outcomes are achieved.

(ii) The arbitrator is entirely unbiased and is unaware of the agents' identities.

(iii) To eliminate the use of randomization merely as a tie-breaking tool, assume randomization is costly. Specifically, a randomization device can be used, which destroys the item with a probability of $c$ and, with a probability of $1-c$, triggers another lottery that assigns the item to either agent with equal probability.

What cost $c$ should we be willing to incur for the sake of randomization? According to both the Nash and Kalai-Smorodinsky solutions, any cost strictly below 1 should be incurred. Therefore, both solutions are willing to sacrifice nearly the entire surplus. According to the RU solution, it is not worth sacrificing any of the surplus. Ultimately, only one of the agents can possess the item, and randomization does not alter this fact.

\begin{appendices}
\section{}\label{ch2:appA}

We begin with a definition of $S_x$ and $P_x$ for a general $n\in\mathbb{N}$. For this purpose, the following notation is introduced. For $x\in[1,n]$, let $\lfloor x \rfloor :=\max\{k\in\{1,...,n\}:k\leq x\}$, meaning $\lfloor x \rfloor$ is $x$ rounded down to the closest integer. For any $u\in\mathbb{R}^n$, let $\pi(u)\subset\mathbb{R}^n$ denote the set containing $u$ and all its permutations. For any $k\in\{0,...,n\}$, let $g(k)$ denote the sequence of length $n$ where $g_i(k)=1$ if $i\leq k$ and $g_i(k)=0$ otherwise. Hence, $g(0)=(0,...,0)$, $g(1)=(1,0,...,0)$ and so on. For any $x\in[1,n]$, let $h(x)$ denote the sequence of length $n$ where $h_i(x)=1$ if $i \leq x$, $h_i(x)=x-\lfloor x \rfloor$ if $i = \lfloor x \rfloor + 1$ and $h_i(k)=0$ otherwise. For any $x\in[1,n]$, let 
\begin{align}\nonumber
P_x &:= \pi(h(x)), & S_x &:= P_x \cup \bigcup_{k=0}^{\lfloor x \rfloor}\pi(g(k)).
\end{align} The following lemma identifies $\conv S_x$ and $\conv P_x$.

\begin{lemma}\label{ch2:biglemma}
For any $x\in[1,n]$, $\conv S_x=\{u\in[0,1]^n:\sum_{i\in N}u_i \leq x\}=:R_x$ and $\conv P_x=\{u\in[0,1]^n:\sum_{i\in N}u_i = x\}=:Q_x$.
\end{lemma}

\begin{proof}
An element $u\in S\subset\mathbb{R}^n$ is an \textit{extreme point} of $S$ if there doesn't exist $v,w\in S$ and $\lambda\in(0,1)$ such that $v\neq w$ and $u=\lambda v + (1-\lambda)w$. In the following, we show that $S_x$ (resp. $P_x$) contains all extreme points of $R_x$ (resp. $Q_x$). Then, by the Krein–Milman theorem, the lemma follows from the fact that $S_x\subseteq R_x$ and $P_x \subseteq Q_x$. 

First, we show that an extreme point $u$ of $R_x$ (resp. $Q_x$) can have at most one coordinate $u_j$ such that $0<u_j<1$. Consider $u\in R_x$ (resp. $Q_x$) with $0<u_j<1$ and $0<u_k<1$ for some $j\neq k$. Then there exists $v,w\in R_x$ (resp. $Q_x$) and $\varepsilon>0$ such that $v_i=w_i=u_i$ whenever $i\notin\{j,k\}$ and \begin{align}
v_j & = u_j + \varepsilon, & w_j & = u_j - \varepsilon,\\
v_k & = u_k - \varepsilon, & w_k & = u_k + \varepsilon.
\nonumber
\end{align}
Since $\frac{1}{2}v+ \frac{1}{2}w = u$, $u$ is not an extreme point of $R_x$ (resp. $Q_x$).

Second, we show that if there is an extreme point of $R_x$ with exactly one coordinate $u_j$ such that $0<u_j<1$, then $u\in\pi(h(x))$. Assume $u\notin\pi(h(x))$ and $0<u_j<1$. Then there exists $v,w\in R_x$ and $\varepsilon>0$ such that $v_i=w_i=u_i$ whenever $i\neq j$ and \begin{align}
v_j & = u_j + \varepsilon, & w_j & = u_j - \varepsilon.
\nonumber
\end{align}
Since $\frac{1}{2}v+ \frac{1}{2}w = u$, $u$ is not an extreme point of $R_x$.

By the above arguments, the only candidates for extreme points of $R_x$ (resp. $Q_x$) are the points in $S_x$ (resp. $P_x$).

\end{proof}

We now present the proof of Theorem \ref{ch2:RU theorem}. The proof is nearly identical to the sketch in Section \ref{ch2:secaxioms}. Nevertheless, we state the proof for the sake of completeness.\\
\ \\
\noindent \textit{Proof of Theorem \ref{ch2:RU theorem}} \nopagebreak[4] \\
First, consider a bargaining set $S_x$ for some $x\in[1,n]$. Note that $P_x$ is the set of Pareto optimal points in $S_x$. Hence, by \ref{ch2:PO}, $f(S_x)\subseteq P_x$. Furthermore, note that $S_x$ is symmetric. So by \ref{ch2:WSYM}, \begin{equation}\label{ch2:neq1}f(S_x)= P_x.\end{equation}

Second, consider the bargaining set $\conv S_x$. By (\ref{ch2:neq1}) and \ref{ch2:NBR}, $P_x \subseteq f(\conv S_x)$. As $\conv S_x$ is convex, $\conv P_x\subseteq f(\conv S_x)$ by \ref{ch2:CS}. From Lemma \ref{ch2:biglemma} we can see that $\conv P_x$ is the set of Pareto-optimal points in $\conv S_x$. Then by \ref{ch2:PO},
\begin{equation}\label{ch2:neq2}f(\conv S_x)=\left\{u\in[0,1]^n:\sum_{i\in N}u_i = x\right\}.\end{equation}

Third, consider any bargaining set $S\in\mathcal{S}$ where $m(S)=(1,...,1)$. Let $x^*:=\max_{u\in S}\sum_{i\in N}u_i$, which must exist due to our assumption of compactness. Note that $S\subseteq \conv S_{x^*}$ and that $S\cap f(\conv S_{x^*})$ is non-empty. Hence, (\ref{ch2:neq2}) and \ref{ch2:WIIA} imply
\begin{equation}\label{ch2:neq3}f(S)=S\cap f(\conv S_{x^*})=\argmax_{u\in S} \sum_{i\in N}u_i.\end{equation}

Fourth, note that for any $S\in\mathcal{S}$ there exists a bargaining set $S' \in\mathcal{S}$ with $m(S')=(1,...,1)$ and a positive linear transformation $\alpha$ such that $\alpha(S)=S'$. Then by (\ref{ch2:neq3}) and \ref{ch2:INV}, \begin{equation}\label{ch2:neq4}f(S)=\argmax_{u\in S}\left(\sum_{i\in N}\frac{u_i}{m_i(S)}\right).\end{equation}
This concludes the proof.

\section{}\label{ch2:app:impl}
This section contains the proof of Theorem \ref{ch2:theoremimpl}. We have already shown in the Section \ref{ch2:sec:game} that for every strictly RU-optimal alternative, there exists an SPE with this alternative as the outcome. Here, we prove the other direction, namely that every SPE outcome is strictly RU-optimal. We fix some $\theta \in \Theta$ and omit it from now on. Let $a$ denote some alternative that isn't strictly RU-optimal.

First, consider Case (1) of the initial stage, where $a$ is implemented through some unanimous proposal $(a,p)$. If $\sum_{i\in N} p_i > \sum_{i\in N} u_i(a)$ then there exists a Player $j$ for whom $p_j>u_j(a)$. This player can deviate to some proposal $(a',p')$ with $\sum_{i\in N} p'_i < \sum_{i\in N} p_i$ and be the first to reject in the approval stage, which gives an expected utility of $p_j$. If $\sum_{i\in N} p_i\leq \sum_{i\in N} u_i(a)$, then there is at least one player who strictly prefers some strictly RU-optimal alternative $a'$. Then this player can deviate to $(a',(u_1(a')-\varepsilon,...,u_n(a')-\varepsilon))$ for $\varepsilon$ sufficiently small such that $\sum_{i\in N} p_i < \sum_{i\in N} \left(u_i(a') -\varepsilon\right)$. Then $a'$ is put to the test in the approval stage and the unique SPE outcome is that all players accept and $a'$ is implemented.

Second, consider Case (2a), where there are two distinct proposals $(a',p')$ and $(a'',p'')$ with $\sum_{i\in N} p'_i = \sum_{i\in N} p''_i$, leading to $a_{\dis}$. If $n\geq 3$, then some Player $i$ can deviate in the initial stage to bring about Case (3) and choose $b_i$. So assume $n=2$. If $p'_1+p'_2 = p''_1+p''_2 > 1$, then a player would be better off proposing a lower sum, be the first to choose in the approval stage and then reject. If $p'_1+p'_2 = p''_1+p''_2 \leq 1$, then for some strictly RU-optimal alternative $a'''$ a player can propose $(a''',(u_1(a''')-\varepsilon,u_2(a''')-\varepsilon))$ for $\varepsilon$ sufficiently small, leading to $a'$.

Third, consider Case (2b), where $a$ is implemented through acceptance of the proposal $(a,p)$ by all players in the approval stage. Unless only a single Player $j$ has made the other proposal $(a',p')$, any player can bring about Case (3) and choose a more preferred alternative. Hence, assume only single Player $j$ has made the other proposal and $a$ is among the best alternatives for all players other than $j$. There must be some strictly RU-optimal alternative $a''$ that is preferred to $a$ by $j$. Since $a$ is unanimously approved, $\sum_{i\in N} p_i \leq \sum_{i\in N} u_i(a)$. Player $j$ can deviate to $(a'',(u_1(a'')-\varepsilon,...,u_n(a'')-\varepsilon))$ for $\varepsilon$ sufficiently small, leading to $a''$.

Fourth, consider Case (2b), where $a$ is chosen after some player rejects in the approval stage. Let $\Sigma:=\sum_{i\in N} u_i(b)$ for any RU-optimal alternative $b$. First, consider $n=2$. Let $(a',p')$ be Player 1's proposal and $(a'',p'')$ be Player 2's proposal. Without loss of generality, assume that $p'_1+p'_2 > p''_1 + p''_2$, such that Player 2 decides first in the approval stage. Player 2 rejects and with probability $p'_2$ chooses $a$ with $u(a)=(x,1)$ for some $x\in[0,1)$. This gives expected utility $p'_2$ to Player 2 and $p'_2 x$ to Player 1. Consider the case where $p''_1 + p''_2 \geq \Sigma$. Player 1 has an incentive to deviate to a proposal with a lower sum unless \begin{equation}\label{ch2:proofth2eq1}p''_1 \leq p'_2 x.\end{equation} Furthermore, since $p''_1 + p''_2 \geq \Sigma$, \begin{equation}\label{ch2:proofth2eq2}p''_1 + p''_2 \geq x + 1.\end{equation} Then (\ref{ch2:proofth2eq1}) and (\ref{ch2:proofth2eq2}) imply $p''_2 = 1$ and either $x=0$ or $p'_2=1$. But $x=0$ is not possible since this would imply $p''_1=0$ and this is not in the strategy-space of $g^*$. Hence $p'_2=1$ and $p''_1=x$. Since we consider the case $p''_1 + p''_2 \geq \Sigma$ and have found $p''=u(a)$, it must be that $a$ is RU-optimal. Furthermore, since $p'_2=1$, $a$ is implemented with certainty. This contradicts the assumption that the SPE outcome is sub-optimal. If $p''_1 + p''_2 < \Sigma$, then Player 1 can implement any strictly RU-optimal alternative $b$ by deviating to the proposal $(b,(u_1(b)-\varepsilon,u_2(b)-\varepsilon))$ for $\varepsilon$ sufficiently small. The only case in which Player 1 would have no incentive to do so is $p'_2=1$ and if $b$ is the best strictly RU-optimal alternative for Player 1, which contradicts the assumption that the SPE outcome is sub-optimal. This concludes the case $n=2$. The argument extends to a general $n$.

Finally, consider Case (3) of the initial stage, where $a$ is the alternative of the proposal with the highest sum strictly below $n$. Then at least one Player $j$ prefers $b_j$ to $a$. This player can deviate to a proposal with an even higher sum below $n$ and the alternative $b_j$.

This concludes the proof of Theorem \ref{ch2:theoremimpl}.

\section{}\label{ch2:app:full}

We follow the notation of \textcite{moore1988subgame}. Consider a two-player extensive game form $g$. Let $T$ denote the set of nodes. For any $t\in T$, we denote by $g(t)$ the sub-game starting at node $t$. For any $t\in T$ and $i\in\{1,2\}$, we denote by $\sigma_i(t)$ the set of actions of Player $i$. We assume that $|\sigma_i(t)|\geq 1$ for all $t\in T$. If $|\sigma_i(t)|= 1$, then Player $i$ has no decision at $t$. If both $|\sigma_1(t)|> 1$ and $|\sigma_2(t)|> 1$, then both players move simultaneously at $t$. Let $\sigma_i:=\bigtimes_{t\in T} \sigma_i(t)$ denote the strategy space of Player $i$ and let $\sigma=\sigma_1 \times \sigma_2$ denote the set of strategy profiles. For any $s\in \sigma$ and $t\in T$, we denote by $s|t$ the part of $s$ that specifies the strategy profile for the game $g(t)$. For any $s\in \sigma$ and $t\in T$, we denote by $s(t)\in \sigma_1(t)\times \sigma_2(t)$ the action pair that $s$ prescribes for the node $t$. Terminal nodes are alternatives in $A$. For any $s\in \sigma$ and $\theta\in\Theta$, $u^\theta(s)=(u_1^\theta(s),u_2^\theta(s))$ denotes the utility vector of the terminal node that results from $s$. For any $s\in \sigma$ and $t\in T$, we write $u^\theta(s|t)$ to denote the utility vector of the terminal node that is reached when starting at $t$ and playing according to $s$.\\
\ \\
\textit{Proof of Proposition \ref{ch2:prop:full}} \nopagebreak[4] \\
Fix some $\theta \in \Theta$ such that $f^{\text{RU}}(S_\theta)$ contains both $(1,0)$ and $(0,1)$. In the following, we omit $\theta$ on the individual utility functions and the game. We have already established in the main section that there must be $s^+,s^-\in \sigma$ such that both are SPE of $g$ with $u(s^+)=(1,0)$ and $u(s^-)=(0,1)$ and that $(s^+_1,s^-_2)=:s^0$ is a NE with $u(s^0)=(0,0)$. In the following, we construct $s^*\in \sigma$ such that $s^*$ is an SPE with $u(s^*)=(0,0)$.

Let $t_0,t_0,...,t_k,t_{k+1}$ denote the nodes of the equilibrium path of $s^0$, where $t_0$ is the initial node of $g$ and $t_{k+1}$ is a terminal node of $g$ associated with the pay-off $(0,0)$. Consider the second to last node $t_k$. Let $(x_k,y_k)\in \sigma_1(t_k) \times \sigma_2(t_k)$ denote the action pair that leads to $t_{k+1}$, formally $t_{k+1} = (t_k,(x_k,y_k))$. If there are only terminal nodes succeeding $t_k$, then $g(t_k)$ is a one-stage game and $s^0|t_k$ is not only a NE of $g(t_k)$ but also an SPE. Hence, $s^*|t_k = s^0|t_k$ ensures that $s^*|t_k$ is an SPE of $g(t_k)$ with outcome $(0,0)$. If there are non-terminal nodes succeeding $t_k$, then we construct $s^*|t_k$ as follows.

Choose $s^*(t_k)=(x_k,y_k)$. For any $(x_k,y)$ with $y\in \sigma_2(t_k)$ choose $s^*|(t_k,(x_k,y)) = s^+|(t_k,(x_k,y))$. This ensures that $s^*|(t_k,(x_k,y))$ is an SPE of $g(t_k,(x_k,y))$ for all $y\in \sigma_2(t_k)$. Furthermore, it must be that $u_2(s^*|(t_k,(x_k,y)))=0$ for all $y\in \sigma_2(t_k)$, because $(t_k,(x_k,y))$ can be reached by a unilateral deviation of Player 2 in the strategy profile $s^+$. Similarly, choose $s^*|(t_k,(x,y_k)) = s^-|(t_k,(x,y_k))$ for all $x\in \sigma_1(t_k)$. For $s^*|(t_k,(x,y))$ such that neither $x=x_k$ nor $y=y_k$, choose an arbitrary SPE. Note that $s^*|t_k$ has been constructed such that an SPE is played at all nodes succeeding $t_k$, such that the outcome is $(0,0)$ and such that no Player has an incentive to deviate at $t_k$. Therefore, $s^*|t_k$ is an SPE of $g(t_k)$ with outcome $(0,0)$.

Finally, consider $t_{l}$ for any $l\in\{0,...,k-1\}$. Let $(x_l,y_l)\in \sigma_1(t_l) \times \sigma_2(t_l)$ denote the decision that leads to $t_{l+1}$, formally $t_{l+1} = (t_l,(x_l,y_l))$. Assume $s^*|t_{l+1}$ is an SPE of $g(t_{l+1})$ with outcome $(0,0)$. Choose $s^*(t_l)=(x_l,y_l)$ and construct $s^*|(t_l,(x,y))$ for $(x,y)\neq (x_l,y_l)$ just as before. Then $s^*|t_l$ is an SPE of $g(t_{l})$ with outcome $(0,0)$. By induction, $s^*$ is an SPE of $g$ with outcome $(0,0)$. This concludes the proof.

\section{}\label{ch2:appD}

We prove the Theorem \ref{ch2:theo:asym} for $n=2$. The proof for general $n$ follows similarly as that of Theorem \ref{ch2:RU theorem}. First, consider the bargaining set $S=\{(0,0),(1,0),(0,1)\}$. By \ref{ch2:PO}, there are three possible cases.

\textit{Case 1:} $f(S)=\{(1,0)\}$

\textit{Case 2:} $f(S)=\{(0,1)\}$

\textit{Case 3:} $f(S)=\{(1,0),(0,1)\}$

\noindent Assume Case 1 holds true. By \ref{ch2:NBR}, $(1,0)\in f(\conv S)$ and by \ref{ch2:WIIA}, $(0,1)\notin f(\conv S)$. By \ref{ch2:C}, there exists a $\lambda\in[0,1]$ such that $f(\conv S \cup \{(\lambda,1)\}) = f(\conv S) \cup \{(\lambda,1)\}$. Note that $\lambda>0$, as it would otherwise contradict the assumption that $(0,1)\notin f(\conv S)$. Furthermore, note that $\lambda<1$ as it would otherwise contradict \ref{ch2:PO}. By \ref{ch2:NBR}, $\{(1,0),(\lambda,1)\}\subseteq f(\conv (S \cup \{(\lambda,1)\}))$ and by \ref{ch2:CS}, \begin{equation}\label{ch2:asym1}f(\conv (S \cup \{(\lambda,1) \}))=\left\{u\in[0,1]^2: u_1+ (1-\lambda) u_2 = 1\right\}.\end{equation} 

Second, consider the bargaining set $$S_x:=\conv \left(S \cup \{(\lambda,1)\}\right)\cup\{(1,x),(x+(1-x)\lambda,1)\}$$ for any $x\in [0,1]$. By \ref{ch2:I} and (\ref{ch2:asym1}), $\{(1,x),(x+(1-x)\lambda,1)\}\subseteq f(S_x)$. By \ref{ch2:NBR}, $\{(1,x),(x+(1-x)\lambda,1)\}\subseteq f(\conv S_x)$ and by \ref{ch2:CS} and \ref{ch2:PO},
\begin{equation}\label{ch2:asym2}
f(\conv S_x)=\left\{u\in[0,1]^2: u_1+ (1-\lambda) u_2 = 1+(1-\lambda)x\right\}.
\end{equation}

Third, consider any bargaining set $S\in\mathcal{S}$ where $m(S)=(1,1)$. Let $x^*:=(1-\lambda)^{-1}\left(\max_{u\in S}(u_1+(1-\lambda)u_2)-1\right)$, which must exist due to our assumption of compactness. Note that $S\subseteq \conv S_{x^*}$ and that $S\cap f(\conv S_{x^*})$ is non-empty. Hence, (\ref{ch2:asym2}) and \ref{ch2:WIIA} imply
\begin{equation}\label{ch2:asym3}f(S)=S\cap f(\conv S_{x^*})=\argmax_{u\in S} (u_1 + (1-\lambda)u_2).\end{equation}

Fourth, note that for any $S\in\mathcal{S}$ there exists a bargaining set $S' \in\mathcal{S}$ with $m(S')=(1,1)$ and a positive linear transformation $\alpha$ such that $\alpha(S)=S'$. Then by (\ref{ch2:asym3}) and \ref{ch2:INV}, \begin{equation}f(S)=\argmax_{u\in S}\left(\frac{u_1}{m_1(S)}+(1-\lambda)\frac{u_2}{m_2(S)}\right).\end{equation}
This concludes Case 1. 

Note that by an analogous argument, we can find an analogous solution for Cases 2 and 3. We can summarize the results for the different cases as follows. There exists $(\mu_1,\mu_2)\in(0,1)^2$ where $\mu_1+\mu_2=1$ such that for any $S\in\mathcal{S}$,
\begin{equation}f(S)=\argmax_{u\in S}\left(\mu_1\frac{u_1}{m_1(S)}+\mu_2\frac{u_2}{m_2(S)}\right).\end{equation}
This concludes the proof.

\end{appendices}

\printbibliography[segment=\therefsegment, heading=subbibintoc]

\end{document}